\documentclass[letterpaper, 10 pt,conference]{ieeeconf}  

\IEEEoverridecommandlockouts                              

\usepackage{graphics} 
\usepackage{epsfig} 
\usepackage{amsmath} 
\usepackage{amssymb}  
\usepackage{url,algorithm}
\usepackage{multirow}
\usepackage{color}
\usepackage{subfigure}
\usepackage{tikz}
\usepackage{pgfplots}
\usepackage{bbm}

\let\labelindent\relax
\usepackage{enumitem}
\renewcommand{\theenumi}{\arabic{enumi}}

\renewcommand{\theenumii}{\arabic{enumii}}

\renewcommand{\theenumiii}{\arabic{enumiii}}

\renewcommand{\theenumiv}{\arabic{enumiv}}

\newtheorem{theorem}{Theorem}

\newtheorem{assumption}{Assumption}
\newtheorem{proposition}{Proposition}
\newtheorem{remark}{Remark}
\newtheorem{problem}{Problem}

\newcommand{\figurename}[1]{Figure{#1}}
\newcommand{\tablename}[1]{Table{#1}}

\title{\LARGE \bf Direct data-driven model-reference control\\with Lyapunov stability guarantees 
}

\author{Valentina Breschi$^{1}$, Claudio De Persis$^{2}$, Simone Formentin$^{1}$ and Pietro Tesi$^{3}$
\thanks{*This project was partially supported by the Italian Ministry of University and Research under the PRIN'17 project \textquotedblleft Data-driven learning of constrained control systems", contract no. 2017J89ARP.}
	\thanks{$^{1}$ Dipartimento di Elettronica, Informatica e Bioingegneria (DEIB), Politecnico di Milano, Piazza Leonardo da Vinci 32, 20133 Milano, Italy. Corresponding author: {\tt\small valentina.breschi@polimi.it}}
	\thanks{$^{2}$ ENTEG, Faculty of Science and Engineering, University of Groningen, 9747 AG Groningen, The Netherlands.}
	\thanks{$^{3}$ Department of Information Engineering (DINFO), University of Florence, 50139 Florence, Italy.}
}

\begin{document}
	
	\maketitle
	\thispagestyle{empty}
	\pagestyle{empty}

	\begin{abstract}
		We introduce a novel data-driven model-reference control design approach for unknown linear systems with fully measurable state. The proposed control action is composed by a static feedback term and a reference tracking block, which are shaped from data to reproduce the desired behavior in closed-loop. By focusing on the case where the reference model and the plant share the same order, we propose an optimal design procedure with Lyapunov stability guarantees, tailored to handle state measurements with additive noise. Two simulation examples are finally illustrated to show the potential of the proposed strategy.
	\end{abstract}
	
\section{Introduction} 

In many control applications, it is usually the case that the desired performance is expressed in terms of an \textit{a-priori specified} closed-loop model, to be matched using a controller with a given structure \cite{nguyen2018model}. 

When a model of the plant to be controlled is available, the above model-matching issue can be formulated as an \textit{interpolation} problem \cite{zhou1996robust}. More often, the mathematical description of the system is not given and thus a model needs to be identified from a set of experimental data 
\cite{sysid}. 
Such a 2-step procedure, namely the sequence of system identification and model-based control design, has been proven to lead to potentially sub-optimal solutions, in that the model that best fits the data is not necessarily also the most suitable for controller tuning \cite{formentin2014comparison}. 

Alternative approaches have thus been proposed to  map the data \textit{directly} onto the controller parameters without undertaking a full modeling study, see, \textit{e.g.,} \cite{hjalmarsson2002iterative,van2011data,batt2018,
formentin2019deterministic}. Such techniques are usually referred to as ``direct data-driven" approaches. Albeit appealing and effective (see, \textit{e.g.,} the successful applications in \cite{campestrini2016unbiased,formentin2012data}), these approaches suffer from 
few drawbacks, which prevent them from being real competitors of more traditional model-based strategies. For instance, data-driven methods are mostly 
conceived to handle SISO (Single Input Single Output)systems, or at most MIMO (Multiple Inputs Multiple Outputs) ones with few input/output channels \cite{formentin2012non}, which make them less suitable 
when the size of the inputs and outputs 
is large. Moreover, stability is guaranteed only asymptotically, i.e.,
as the number of data goes to infinity, see,
e.g., \cite{batt2018,formentin2012data,van2011data}.

In this work, we propose a new direct data-driven design strategy for model-reference control, 
which is endowed with stability guarantees. 
The key technical step is the data-based 
representation of linear systems originally proposed in \cite{WILLEMS2005} within the behavioral 
framework, and recently reconsidered using state-space representations
in \cite{DePersis2020}. 
This data-based representation
has emerged as a powerful tool for analizing the properties of dynamic systems, \emph{e.g.,} to retrieve the passivity indexes of a system \cite{Romer2019,Romer2019bis,Tanemura2019}, and for tackling many important
control problems, such as 
predictive 
\cite{Coulson2018,berberich2019MPC,Allibhoy2021}, optimal and robust control 
\cite{DePersis2020,berberich2019robust,vanwaarde2020noisy,Matni2020,baggio2020}, as well as control of time-varying and nonlinear systems \cite{DePersis2020,Guo2020,Nortmann2020}. 
Inspired by \cite{DePersis2020}, we tackle the model-reference control problem in a state-space setting. For the case of noise-free data, we show that this problem can be {equivalently} cast as  a semi-definite program, and thus solved using efficient optimization solvers. We also account for possible model mismatches by considering a regularization-based strategy that moves the matching constraints to the objective function (see \cite{dorf2020} for an excellent discussion on the use of regularization techniques in the context of data-driven control). In the spirit of \cite{deptLQR}, we finally address the case of noisy measurements by considering a strategy based on averaging multiple experiments and derive sufficient conditions for closed-loop stability. 
This strategy is tested on two different numerical examples through Monte Carlo simulations.
Throughout the paper, we will also discuss 
some variants of the baseline solution which can be adopted to increase the degree of
robustness to noise and to find approximate solutions when the choice of the reference
model results in an unfeasible problem.

The remainder of the paper is as follows. In Section \ref{Sec:setting}, the model-reference 
control design problem is formally stated. Section \ref{Sec:DD_closedloop} describes the 
key technical passages allowing us to formulate the entire problem as a data-based optimization task. 
The main results are presented 
in Section \ref{Sec:matching_prob} and 
Section \ref{Sec:noise}, and some examples are discussed in Section \ref{Sec:examples}.
Section \ref{sec:conclusions} ends the paper with concluding remarks.


\section{Setting and problem formulation}\label{Sec:setting} 
 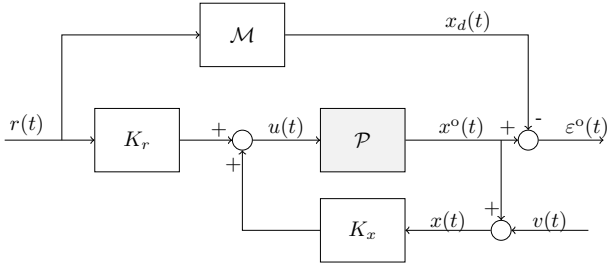
\begin{figure}[!tb]
 	\centering
 	\scalebox{0.8}{\begin{tikzpicture}
 		\node [coordinate] (reference){};
 		\node [coordinate,node distance=.95cm, right of=reference] (aid1){};
 		\node [draw, rectangle, 
 		minimum height=3em, minimum width=4em, node distance=1.25cm, right of=aid1] (Kr) {$K_{r}$};
 		\node [draw, circle, node distance=1.75cm, right of=Kr] (sum) {};
 		\node [draw, rectangle,  node distance=2cm,
 		minimum height=3em, minimum width=4em, right of=sum,fill=black!5] (system) {$\mathcal{P}$};
 		 \node [draw, rectangle,  node distance=1.5cm,minimum height=3em, minimum width=4em, below of=system,] (Kx) {$K_x$};
 		\node [draw, circle, node distance=2.75cm, right of=system] (sum3) {};
 		\node [coordinate,node distance = 0.45cm, left of =sum3] (aid2) {};
 		\node [draw,circle,node distance = 1.5cm, below of =aid2] (sum4) {};
 		\node [coordinate, right of=sum3, node distance=1.25cm] (output) {};
 		\node [coordinate,right of = sum3,node distance=1cm] (aid6) {};
 		\node [coordinate,below of = aid6,node distance=1.5cm] (noise) {};
 		\node [draw, rectangle,minimum height=3em, minimum width=4em, above of=sum, node distance=1.75cm] (M) {$\mathcal{M}$};
 		
 		\draw [draw,->] (reference) -- node [yshift=0.25cm,near start] {$r(t)$} (Kr);
 		\draw [draw,->] (Kr) -- node [near end,yshift=0.175cm] {$\tiny{+}$} (sum);
		\draw [draw,->] (aid1) |- node {} (M);
 		\draw [draw,->] (sum) -- node [yshift=0.2cm] {$u(t)$} (system);
 		\draw [draw,->] (aid2) -- node[near end,xshift=-.15cm,yshift=-0.15cm] {$\tiny{+}$} (sum4);
 		\draw [draw,->] (sum4) -- node[yshift=0.15cm] {$x(t)$} (Kx);
 		\draw [draw,->] (Kx) -| node [near end,yshift=0.425cm,xshift =-0.15cm] {$\tiny{+}$} (sum);
		\draw[draw,->] (system) -- node [yshift=0.2cm] {$x^{\mathrm{o}}(t)$} node [near end,yshift=.2cm,xshift=.3cm] {$\tiny{+}$} (sum3);
		\draw [draw,->] (M) -| node [near end,xshift=0.175cm,yshift=-.65cm] {-} node [xshift=-1cm,yshift=.225cm] {$x_{d}(t)$} (sum3);
		\draw [draw,->] (sum3) -- node[yshift=0.2cm,near end]{$\varepsilon^{\mathrm{o}}(t)$}(output);
		\draw[->] (noise)--node[yshift=0.15cm]{$v(t)$}(sum4);
 		\end{tikzpicture}}	 \vspace{-.1cm}
 	\caption{Matching scheme. The unknown plant $\mathcal{P}$ is highlighted in gray. The mismatch error is denoted with $\varepsilon^{\mathrm{o}}(t)$.
 	}\label{Fig:dd_matchingscheme}
 \end{figure}

Let $\mathcal{P}$ be a discrete-time \emph{linear time invariant} (LTI) system, the evolution of which is dictated by the following equations:
\begin{subequations}\label{eq:noiseless_model}
\begin{equation}
\mathcal{P}:\quad \begin{cases}
x^{\mathrm{o}}(t+1)=Ax^{\mathrm{o}}(t)+Bu(t),\\
x(t) = x^{\mathrm{o}}(t) + v(t),
\end{cases}
\end{equation}
\end{subequations}
where $x^{\mathrm{o}}(t) \in \mathbb{R}^{n}$ is the state, 
$u(t) \in \mathbb{R}^{m}$ is an exogenous input signal fed into the system at time $t \in \mathbb{N}$, and we measure the noisy state $x(t) \in \mathbb{R}^{n}$, 
where $v(t) \in \mathbb{R}^{n}$ 
represents {measurement noise}.

We aim at designing a \emph{stabilizing} controller that matches 
a user-defined reference model $\mathcal{M}$, dictating 
the desired 
closed-loop response to a customizable set point $r(t) \in \mathbb{R}^{n}$. 
The reference model is 
characterized by the equation
\begin{equation}\label{eq:desired_behavior}
\mathcal{M}: \quad 
x_{d}(t+1)=A_{M}x_{d}(t)+B_{M}r(t)
\end{equation} 
with $x_{d}(t) \in \mathbb{R}^{n}$ indicating the desired state at time $t$.
Here, the reference model $\mathcal{M}$ is assumed to be 
stable and provided at design time, thus its matrices 
are fixed and known. The matching problem can 
be formally stated as follows.

\begin{problem}[Matching problem]\label{pb:MMstatic_problem}
	Let $\mathcal{P}$ be an LTI system as in \eqref{eq:noiseless_model} and  $\mathcal{M}$ be a stable reference model as in \eqref{eq:desired_behavior}.
	The matching problem amounts to finding two control matrices 
	$K_{x}, K_{r} \in \mathbb{R}^{m \times n}$ such that
\begin{subequations}\label{eq:MB_matching}
	\begin{align}
	&A+BK_{x}=A_{M}, \\
	&BK_{r}=B_{M},
	\end{align}
\end{subequations}
If such matrices 
exist, 
we say that the matching problem is \emph{feasible}. \hfill $\square$
\end{problem} 
\indent This definition 
comes from the fact that, if the equations in \eqref{eq:MB_matching} have a solution, then the control law
 \begin{equation}\label{eq:noisy_law}
u(t)=K_{x}x(t)+K_{r}r(t),
\end{equation} 
ensures that the noiseless behavior of $\mathcal{P}$ matches that of $\mathcal{M}$. Note that $K_{x}$ is responsible for closed-loop stability while $K_{r}$ is a feed-forward term allows one to attain 
the desired response to 
$r(t)$. The considered matching scheme is depicted in the
block diagram of \figurename{~\ref{Fig:dd_matchingscheme}}. 

In principle, solving Problem~\ref{pb:MMstatic_problem} requires the knowledge of the plant matrices $(A,B)$ characterizing \eqref{eq:noiseless_model}. In this paper, we are interested in solving the matching problem when $A$ and $B$ are \emph{unknown}, and we have access to a finite set of input-state pairs only. The available data are obtained by applying an input signal $\mathcal{U}_{T-1}=\{u(t)\}_{t=0}^{T-1}$ to $\mathcal{P}$ and measuring the 
response of the system $\mathcal{X}_{T}=\{x(t)\}_{t=0}^{T}$, where $T$ denotes the length of the experiment\footnote{Open-loop experiments can be carried out on stable systems, while a stabilizing controller is assumed to be available to perform closed-loop experiments when the plant is unstable.}. This new matching problem is formalized as follows.

\begin{problem}[Data-driven matching problem]\label{pb:DDCstatic_problem}
	Let $\mathcal{P}$ be an LTI system as in \eqref{eq:noiseless_model}, and 
	let $\mathcal{M}$ be a stable reference model as in \eqref{eq:desired_behavior}.
	The data-driven matching problem amounts to finding two control matrices 
	$K_{x}, K_{r} \in \mathbb{R}^{m \times n}$ satisfying the conditions in \eqref{eq:MB_matching} by using a set of data $(\mathcal{U}_{T-1},\mathcal{X}_{T})$.
	\hfill $\square$
\end{problem} 
 
\section{Data-based description of the closed-loop}\label{Sec:DD_closedloop}
Consider an experiment of length $T$ carried out on $\mathcal P$, define the following data matrices:
\begin{subequations}\label{eq:aid_matrices}
	\begin{align}
	X_{0,T-1}&=\begin{bmatrix}
	x(0) & x(1) & \cdots & x(T-1)
	\end{bmatrix}, \label{eq:X0T1}\\
	U_{0,T-1}&=\begin{bmatrix}
	u(0) & u(1) & \cdots & u(T-1)
	\end{bmatrix},\label{eq:U0T1}\\
	V_{0,T-1}&=\begin{bmatrix}
	v(0) & v(1) & \cdots & v(T-1)
	\end{bmatrix},\label{eq:V0T1}\\
	X_{1,T}&=\begin{bmatrix}
	x(1) & x(2) & \cdots & x(T)
	\end{bmatrix},\label{eq:X1T}\\
	V_{1,T}&=\begin{bmatrix}
	v(1) & v(2) & \cdots & v(T)
	\end{bmatrix}, \label{eq:V1T} \vspace{-.1cm}
	\end{align}
\end{subequations}  
and let $X_{0,T-1}^{\mathrm{o}}$ and $X_{1,T}^{\mathrm{o}}$ be
the noiseless counterparts of the matrices in \eqref{eq:X0T1} and \eqref{eq:X1T}, \emph{i.e.,}
\begin{subequations}\label{eq:aid_noiseless}
	\begin{align}
	X_{0,T-1}^{\mathrm{o}}&=\begin{bmatrix}
	x^{\mathrm{o}}(0) & x^{\mathrm{o}}(1) & \cdots & x^{\mathrm{o}}(T-1)
	\end{bmatrix}, \label{eq:X0T1_noiseless}\\
	X_{1,T}^{\mathrm{o}}&=\begin{bmatrix}
	x^{\mathrm{o}}(1) & x^{\mathrm{o}}(2) & \cdots & x^{\mathrm{o}}(T)
	\end{bmatrix}. \label{eq:X1T_noiseless} \vspace{-.1cm}
	\end{align}
\end{subequations}
Consider now the following assumption, which is related to the richness of the data.
\begin{assumption} \label{ass:rank_condition}
	The following condition holds:
	\begin{equation}\label{eq:rank_condition}
	\mbox{rank}\left( 
	\begin{bmatrix}
	U_{0,T-1}\\ X_{0,T-1}
	\end{bmatrix}
	\right)=n+m. \vspace{-.25cm}
	\end{equation}
	\hfill $\square$
\end{assumption}
As shown 
next
,
condition \eqref{eq:rank_condition} makes it possible 
to express the behavior of $\mathcal P$
in feedback with \eqref{eq:noisy_law}
purely in terms of the data matrices in \eqref{eq:aid_matrices} for \emph{any} 
control gains $K_x$ and $K_r$. For controllable systems and noiseless 
data this condition 
can be enforced at the design stage by
choosing $\mathcal{U}_{T-1}$ as a \emph{persistently exciting} 
signal \cite{WILLEMS2005}. 
It is also simple to see that Assumption \ref{ass:rank_condition} can be verified from data and that it holds even with noisy ones, whenever the noise is sufficiently small in 
magnitude. 

\begin{proposition}[Data-driven closed-loop representation]
\label{prop:DDCL_representation}
	Let Assumption \ref{ass:rank_condition} be satisfied.
	For any matrices $K_x,K_r$, the closed-loop dynamics resulting from
	the control law \eqref{eq:noisy_law} can be equivalently expressed 
	in terms of data as 
		\begin{equation} \label{eq:DD_CL}
		x^{\mathrm{o}}(t+1)=
		{A}_{cl}x^{\mathrm{o}}(t)+{B}_{cl}r(t)+{D}_{cl}v(t)
		\end{equation}
		with ${A}_{cl}=(X_{1,T}+W_{0,T})G^{x}$, 
		${B}_{cl}=(X_{1,T}+W_{0,T})G^{r}$ and 
		${D}_{cl}=(X_{1,T}+W_{0,T})G^{v}$
	where the matrices
	$G^{x}, G^{v}$ and $G^{r}$ 
	satisfy the equalities
	 \begin{align}
	 &\begin{bmatrix}
	 K_{x} \\ I_{n}
	 \end{bmatrix}=\begin{bmatrix}
	 U_{0,T-1}\\ X_{0,T-1}
	 \end{bmatrix}G^{x}, \label{eq:Gx_cond}\\
	 &\begin{bmatrix}
	 K_{r} \\ \mathbf{0}_{n}
	 \end{bmatrix}=\begin{bmatrix}
	 U_{0,T-1}\\ X_{0,T-1}
	 \end{bmatrix}G^{r}, \label{eq:Gr_cond}\\
	 &\begin{bmatrix}
	 K_{x} \\ \mathbf{0}_{n}
	 \end{bmatrix}=\begin{bmatrix}
	 U_{0,T-1}\\ X_{0,T-1} 
	 \end{bmatrix}G^{v}, \label{eq:Gv_cond}
	 \end{align}
	 and where 
	 $W_{0,T}=AV_{0,T-1}-V_{1,T}$.
%
	\hfill $\square$
\end{proposition}

\begin{proof}
	By combining the dynamics in \eqref{eq:noiseless_model} with
	\eqref{eq:noisy_law}, 
	the closed-loop dynamics is given by
	\begin{equation}\label{eq:noisy_CL}
	x^{\mathrm{o}}(t+1)=(A+BK_{x})x^{\mathrm{o}}(t)+BK_{r}r(t)+BK_{x}v(t),
	\end{equation} 
	which can be equivalently written as:
	\begin{equation}\label{eq:equivalent_CL}
	x^{\mathrm{o}}(t+1)\!=\!\begin{bmatrix}
	B \!&\! A
	\end{bmatrix}\left\{
	\begin{bmatrix}
	K_{x} \\ I_{n}
	\end{bmatrix}\!x^{\mathrm{o}}(t)\!+\!\begin{bmatrix}
	K_{r} \\ \mathbf{0}_{n}
	\end{bmatrix}\!r(t)\!+\!\begin{bmatrix}
	K_{x} \\ \mathbf{0}_{n}
	\end{bmatrix}\!v(t)
	\right\}\!.
	\end{equation}
	Since the rank condition in \eqref{eq:rank_condition} holds, 
	by the Rouch\'e-Capelli theorem there exist three matrices 
	$G^{x}, G^{r}, G^{v} \in \mathbb{R}^{T \times n}$ 
	such that \eqref{eq:Gx_cond}-\eqref{eq:Gv_cond} are satisfied. 
	Therefore, the following equalities hold:
	\begin{subequations}\label{eq:DD_equalities1}
		\begin{align}
		\nonumber \begin{bmatrix}
		B \!\!&\!\! A
		\end{bmatrix}\begin{bmatrix}
		K_{x} \\ I_{n}
		\end{bmatrix} & \!=\! \begin{bmatrix}
		B \!\!&\!\! A
		\end{bmatrix}\begin{bmatrix}	
		U_{0,T-1}\\ X_{0,T-1}
		\end{bmatrix}G^{x}\\
		\nonumber & \!=\!\begin{bmatrix}
		B \!\!&\!\! A
		\end{bmatrix}\begin{bmatrix}	
		U_{0,T-1}\\ X_{0,T-1}^{\mathrm{o}}\!+\!V_{0,T-1}
		\end{bmatrix}G^{x}\\
		&\!=\!\begin{bmatrix}
		B \!\!&\!\! A
		\end{bmatrix}\begin{bmatrix}	
		U_{0,T-1}\\ X_{0,T-1}^{\mathrm{o}}
		\end{bmatrix}G^{x}\!+\!\!AV_{0,T\!-\!1}G^{x}, 
		\label{eq:DD_equalityx1}\\
		 \begin{bmatrix}
		B \!\!&\!\! A
		\end{bmatrix}\begin{bmatrix}
		K_{r} \\ \mathbf{0}_{n}
		\end{bmatrix} &\!=\! \begin{bmatrix}
		B \!\!&\!\! A
		\end{bmatrix}\begin{bmatrix}	
		U_{0,T-1}\\ X_{0,T-1}^{\mathrm{o}}
		\end{bmatrix}G^{r}\!+\!\!AV_{0,T-1}G^{r}, 		
		\label{eq:DD_equalityr1}\\
	    \begin{bmatrix}
		B \!\!&\!\! A
		\end{bmatrix}\begin{bmatrix}
		K_{x} \\ \mathbf{0}_{n}
		\end{bmatrix} &\!=\! \begin{bmatrix}
		B \!\!&\!\! A
		\end{bmatrix}\begin{bmatrix}	
		U_{0,T-1}\\ X_{0,T-1}^{\mathrm{o}}
		\end{bmatrix}G^{v}\!+\!\!AV_{0,T\!-\!1}G^{v}. 		
		\label{eq:DD_equalityv1}
		\end{align} 
	\end{subequations}
	Because of the dynamics in \eqref{eq:noiseless_model}, 
	it straightforwardly follows that 
	\begin{equation*}
	\begin{bmatrix}
	B \!&\! A
	\end{bmatrix}\begin{bmatrix}	
	U_{0,T-1}\\ X_{0,T-1}^{\mathrm{o}}
	\end{bmatrix}=X_{1,T}^{\mathrm{o}},
	\end{equation*}
	which, in turn, implies that \eqref{eq:DD_equalities1} is equivalent to:
	\begin{subequations}\label{eq:DD_equalities2}
		\begin{align}
	 \begin{bmatrix}
	B \!&\! A
	\end{bmatrix}\begin{bmatrix}
	K_{x} \\ I_{n}
	\end{bmatrix} & \!=\! X_{1,T}^{\mathrm{o}}G^{x}\!+
	\!AV_{0,T\!-\!1}G^{x}, \label{eq:DD_equalityx2}\\
	\begin{bmatrix}
	B \!&\! A
	\end{bmatrix}\begin{bmatrix}
	K_{r} \\ \mathbf{0}_{n}
	\end{bmatrix} &\!=\! X_{1,T}^{\mathrm{o}}G^{r}\!+\!AV_{0,T-1}G^{r}, 		\label{eq:DD_equalityr2}\\
	\begin{bmatrix}
	B \!&\! A
	\end{bmatrix}\begin{bmatrix}
	K_{x} \\ \mathbf{0}_{n}
	\end{bmatrix} &\!=\! X_{1,T}^{\mathrm{o}}G^{v}\!+\!AV_{0,T\!-\!1}G^{v}. 		\label{eq:DD_equalityv2}
	\end{align} 
	\end{subequations}	
Finally, since $X_{1,T}=X_{1,T}^{\mathrm{o}}+V_{1,T}$ we obtain
	\begin{subequations}\label{eq:DD_equalities3}
		\begin{align}
		\begin{bmatrix}
		B \!&\! A
		\end{bmatrix}\begin{bmatrix}
		K_{x} \\ I_{n}
		\end{bmatrix} & \!=\! \left(\!X_{1,T}+W_{0,T}\!\right)G^{x}, \label{eq:DD_equalityx3}\\
		\begin{bmatrix}
		B \!&\! A
		\end{bmatrix}\begin{bmatrix}
		K_{r} \\ \mathbf{0}_{n}
		\end{bmatrix} &\!=\! \left(\!X_{1,T}+W_{0,T}\!\right)G^{r}, 		\label{eq:DD_equalityr3}\\
		\begin{bmatrix}
		B \!&\! A
		\end{bmatrix}\begin{bmatrix}
		K_{x} \\ \mathbf{0}_{n}
		\end{bmatrix} &\!=\! \left(\!X_{1,T}+W_{0,T}\!\right)G^{v}, 		\label{eq:DD_equalityv3}
		\end{align} 
	\end{subequations}
	 which easily leads to the data-based representation in \eqref{eq:DD_CL}. 
\end{proof}


\begin{remark}[Relaxing Assumption \ref{ass:rank_condition}] \label{rem:relaxing}
	Having $X_{0,T-1}$ full row rank is necessary to have \eqref{eq:Gx_cond}-\eqref{eq:Gv_cond} fulfilled. In contrast, for some $K_x$ and $K_r$, \eqref{eq:Gx_cond}-\eqref{eq:Gv_cond} might have a solution even when $U_{0,T-1}$ is not full row rank, in line with what has been shown in \cite{vanwaarde2020}. Nonetheless, Assumption \ref{ass:rank_condition} ensures that the data-based representation of 
	Proposition \ref{prop:DDCL_representation} is valid for \emph{any} 
	control matrices $K_x,K_r$.
	\hfill $\blacksquare$ 
\end{remark}

By Proposition \ref{prop:DDCL_representation}, the design problem
can thus be cast in a data-driven fashion as 
 \begin{subequations}\label{eq:DD_matching}
 	\begin{align}
 	&(X_{1,T}+W_{0,T})G^{x}=A_{M}, \label{eq:DD_matchingx}\\
 	&(X_{1,T}+W_{0,T})G^{r}=B_{M}, \label{eq:DD_matchingr} 
 	\end{align}
 which have to be paired with the two consistency conditions 
 \begin{align}
 &X_{0,T-1}G^{x}=I_{n}, \label{eq:DD_matchingx_consistency} \\
 &X_{0,T-1}G^{r}=\mathbf{0}_{n}, \label{eq:DD_matchingr_consistency}
 \end{align}
 \end{subequations}
needed for the constraints \eqref{eq:DD_matchingx} and 
\eqref{eq:DD_matchingr} to be equivalent 
to the model-based ones in \eqref{eq:MB_matching}. If the matching problem is feasible, $K_x$ and $K_r$  can be retrieved a posteriori using the relations $K_x=U_{0,T-1}G^{x}$ and $K_r=U_{0,T-1}G^{r}$, while $G^{v}$ is not explicitly needed for computing $K_x$ and $K_r$. We summarize this fact in the following result.
%
\begin{theorem} \label{thm:equivalence}
Consider system \eqref{eq:noiseless_model} along with a 
reference model as in \eqref{eq:desired_behavior}. Suppose that 
Assumption \ref{ass:rank_condition} holds. 
Then, the matching problem is feasible, namely there exist two control
matrices $K_x$ and $K_r$ satisfying \eqref{eq:MB_matching}, if and only if there exist
matrices $G^{x}$ and $G^{r}$ such that
\eqref{eq:DD_matching} holds.
In such a case, $K_x$ and $K_r$ are given by $K_x=U_{0,T-1}G^{x}$
and $K_r=U_{0,T-1}G^{r}$, respectively. \hfill $\blacksquare$ 
\end{theorem}

\section{Data-driven model matching}\label{Sec:matching_prob} 
Theorem \ref{thm:equivalence} provides an equivalent data-based formulation of the 
model-based matching problem without any explicit identification step. When the data 
are noiseless, the conditions in \eqref{eq:DD_matching} reduce to 
\begin{subequations}\label{eq:DD_matching_noiseless}
	\begin{align}
	&X_{1,T}G^{x}=A_{M}, \label{eq:DD_matching_noiselessAm}\\
	&X_{1,T}G^{r}=B_{M}, \label{eq:DD_matching_noiselessBm} \\
	&X_{0,T-1}G^{x}=I_{n}, \label{eq:DD_matching_noiselessGx} \\
	&X_{0,T-1}G^{r}=\mathbf{0}_{n}. \label{eq:DD_matching_noiselessGr}
	\end{align}
\end{subequations}
Therefore, in this case, Theorem \ref{thm:equivalence} gives a complete and easy-to-implement data-based solution for the matching problem. We summarize this result next, showing that Problems~\ref{pb:MMstatic_problem} and \ref{pb:DDCstatic_problem} are indeed equivalent with noiseless data satisfying Assumption \ref{ass:rank_condition}.
%
%
\begin{theorem} \label{thm:noiseless}
Consider system \eqref{eq:noiseless_model} along with a 
reference model as in \eqref{eq:desired_behavior}. Let 
Assumption \ref{ass:rank_condition} hold for a set of noise-free data. Then, the
matching problem is feasible if and only if \eqref{eq:DD_matching_noiseless} is satisfied. In this case,
any solution $(G^{x},G^{r})$ is such that the control matrices $K_x=U_{0,T-1}G^{x}$ and $K_r=U_{0,T-1}G^{r}$ solve the matching problem. \hfill $\blacksquare$ 
\end{theorem}

\begin{remark}
In the noiseless case, any solution of the system of equations \eqref{eq:DD_matching_noiseless} ensures 
the stability of the closed-loop system. Indeed, matrix $A_M$ is stable by hypothesis and 
$X_{1,T}G^{x}=A+BK_x$.
\hfill $\blacksquare$ 
\end{remark}


\subsection{Handling model mismatch}

Theorem \ref{thm:noiseless} rests on the assumption that the matching 
problem is feasible. In case one has selected a reference model
$\mathcal M$ for which perfect matching is not possible, \eqref{eq:DD_matching_noiseless} 
will have no solutions. One way to remedy this situation is to modify 
\eqref{eq:DD_matching_noiseless} so as to search for a stabilizing controller that
best matches $\mathcal M$ in some suitable sense, as specified below.

An intuitive way for relaxing the matching constraints is 
to lift them to an objective function, recasting the matching problem as
\begin{eqnarray} \label{eq:relax0}
\begin{array}{rl}
\underset{G^x,G^r}{\text{minimize}} & 
\|{X}_{1,T}G^{x}\!-\!A_{M}\| + \lambda \|{X}_{1,T}G^{r}\!-\!B_{M}\| \\[0.1cm]
\text{subject to} & \eqref{eq:DD_matching_noiselessGx}, \eqref{eq:DD_matching_noiselessGr}
\end{array}
\end{eqnarray}
where 
$\lambda>0$ 
weights the relative 
importance between the two matching objectives and 
$\|\cdot\|$ is any norm. 

With this formulation, 
we have guarantees that the solution returns a stabilizing controller $K_x$ only
if the matching problem is feasible
. Nonetheless, even if the reference model is not perfectly matched, it is not difficult to incorporate a stability constraint in this new formulation.
As a first step, note that the condition \eqref{eq:DD_matching_noiseless} can be equivalently written as
\begin{subequations}\label{eq:DD_matching_noiseless2}
	\begin{align}
	&X_{1,T}Q^{x}=A_{M} P, \label{eq:DD_matchinq_noiselessAm2}\\
	&X_{1,T} Q^{r}=B_{M} P, \label{eq:DD_matchinq_noiselessBm2}\\
	 &X_{0,T-1}Q^{x}=P, \label{eq:DD_matching_noiselessQx}\\
	&X_{0,T-1}Q^{r}=\mathbf{0}_{n}, \label{eq:DD_matching_noiselessQr}
	\end{align}
\end{subequations}
having defined $Q^x=G^xP$  and $Q^r=G^rP$, where $P \succ 0$ but otherwise arbitrary.
As a second step, recall that in the noiseless case $X_1 G^x = A+BK_x$, which can be rewritten as $X_1 Q^x P^{-1} = A+BK_x$. Hence, the closed-loop system with feedback controller $K_x$ is stable 
if and only if there exists a matrix $P \succ 0$ that satisfies the Lyapunov inequality
$X_{1,T}Q^{x}P^{-1} (X_{1,T}Q^{x})'-P\prec 0$ which, in turn, can be rewritten as
\begin{equation}\label{eq:LMI_stabcond}
\begin{bmatrix}
P & X_{1,T}Q^{x}\\
(X_{1,T} Q^{x})' & P
\end{bmatrix} \succ \mathbf{0}_{2\cdot n},
\end{equation}
by using Schur complement. This immediately leads to 
the following formulation:
\begin{eqnarray} \label{eq:relax}
\begin{array}{rl}
\displaystyle \underset{Q^x,Q^r,P}{\text{minimize}} & 
\|{X}_{1,T}Q^{x}\!-\!A_{M} P\| + \lambda \|{X}_{1,T}Q^{r}\!-\!B_{M}P\| \\[0.1cm]
\text{subject to} & \eqref{eq:DD_matching_noiselessQx}, \eqref{eq:DD_matching_noiselessQr},
\eqref{eq:LMI_stabcond}
\end{array}
\end{eqnarray}
where 
$\lambda>0$. Note that \eqref{eq:LMI_stabcond} ensures $P \succ 0$. Compared with \eqref{eq:relax0}, this new formulation is such that any solution returns a stabilizing controller. We remark that the reason to work with $Q^x,Q^r$ instead of $G^x,G^r$ is to retrieve a formulation that can be efficiently solved numerically. Indeed, by using $Q^x,Q^r$, 
the matching problem in \eqref{eq:relax} 
corresponds to a \emph{semi-definite program}, that can be efficiently handled by many existing solvers. The specific properties of this formulation are summarized in the next theorem.
%
%
%
%
%
\begin{theorem} \label{thm:mismatch}
Consider system \eqref{eq:noiseless_model} along with a 
reference model as in \eqref{eq:desired_behavior}. Let the data be gathered over a noise-free experiment 
and let Assumption \ref{ass:rank_condition} hold. Then:
\begin{enumerate}
\item[(i)] If there exists a stabilizing static state-feedback linear controller for \eqref{eq:noiseless_model}, 
the program \eqref{eq:relax}
is feasible and any solution $(Q^{x},Q^{r},P)$ is such that 
$K_x=U_{0,T-1}Q^{x}P^{-1}$ ensures closed-loop stability.
\item[(ii)] If the matching problem is feasible, the program \eqref{eq:relax}
is also feasible and any solution $(Q^{x},Q^{r},P)$ is such that 
$K_x=U_{0,T-1}Q^{x}P^{-1}$ and $K_r=U_{0,T-1}Q^{r} P^{-1}$
solve the matching problem.
\end{enumerate}
\end{theorem}

\begin{proof}
(i) Let $\underline K_x$ be any stabilizing controller
and let $\underline K_r$ be an arbitrary matrix of dimension $m \times n$.
By Assumption \ref{ass:rank_condition}, there exists a matrix $G^x$ such that 
$X_{0,T-1}G^x=I_n$ and $U_{0,T-1}G^x=\underline K_x$. This implies
$A+B \underline K_x = X_{1,T}G^x$. Since $\underline K_x$ is stabilizing, there exists a matrix $P \succ 0$ such that $X_{1,T}G^{x}P (X_{1,T}G^{x})'-P\prec 0$. Further, there exists a matrix $G^r$ such that $X_{0,T-1}G^r=\mathbf{0}_n$ and $U_{0,T-1}G^r=\underline K_r$. Hence, $Q^x=G^x P$ and $Q^r=G^r P$ ensure the fulfillment of \eqref{eq:DD_matching_noiselessQx}, \eqref{eq:DD_matching_noiselessQr} and \eqref{eq:LMI_stabcond}, meaning that \eqref{eq:relax} is feasible. Let now $Q^x,Q^r,P\succ0$ be any solution to \eqref{eq:relax}. 
In view of the constraint \eqref{eq:DD_matching_noiselessQx} and
because $K_x=U_{0,T-1}Q^{x}P^{-1}$, we have 
$X_{1,T}Q^{x}P^{-1} = (A X_{0,T-1} + B U_{0,T-1})Q^{x}P^{-1} = A+B K_x$ and, thus, \eqref{eq:LMI_stabcond} ensures that $K_x$ is stabilizing. (ii) The proof of the 
second part of the theorem follows directly from the previous point and the fact that, in this case, the minimum of the objective function is zero. 
\end{proof}
\begin{remark}[Accounting for DC gains] The cost in \eqref{eq:relax} can me modified to account for other properties. As an example, we can add a term in the cost function penalizing $\|P-X_{1,T}Q^{x}-X_{1,T}Q^{r}\|$ to steer the closed-loop DC gain to be unitary, if $\mathcal{M}$ has been chosen so as to have prefect tracking of step-like references. 
\hfill $\blacksquare$
\end{remark}

\section{Noise-aware data-based matching strategy}\label{Sec:noise}
With data corrupted by noise, the analysis gets sensibly more complex. Indeed, in this scenario, the solution to \eqref{eq:relax} might not exist or it may lead to a controller $K_x$ that is not stabilizing. In this section, we provide a first strategy to handle noisy data.
%

Suppose that the noise samples have zero-mean and are \emph{independent} and \emph{identically distributed} (i.i.d.). By the \emph{Strong Law of Large Numbers} it holds
\begin{equation}
\lim_{N \longrightarrow \infty} \frac{v(0)+v(1)+\ldots+v(N)}{N}=0.
\end{equation}
with probability 1. This simple, yet fundamental, property 
suggests that if we perform multiple experiments on the system 
and we then average the collected data, the effect of noise will 
become decreasingly marked as the number of experiment increases. 
In this light, suppose that we make $N$ experiments on system 
\eqref{eq:noiseless_model}, each of length $T$. Let 
$(U_{0,T-1}^{(i)},X_{0,T-1}^{(i)},X_{1,T}^{(i)})$, $i=1,\ldots,N$, 
be the data matrices resulting from the $i$-th experiment and 
$(V_{0,T-1}^{(i)},V_{1,T}^{(i)})$ be the corresponding noise matrices. Denote with
\begin{eqnarray}
\bar S := \frac{1}{N} \sum_{i=1}^N S^{(i)} \vspace{-.1cm}
\end{eqnarray} 
the average of $N$ matrices $S^{(i)}$\footnote{$S^{(i)}$ is used as a place-holder to indicate any data-based matrix contructed based on the $i$-th experiment.}. From a practical viewpoint and in light of the considerations made for the noiseless case, we can then cast the following optimization problem:
\begin{subequations} \label{eq:relax1}
\begin{align}
& \underset{Q^x,Q^r,P}{\text{minimize} } 
~~\|\bar {X}_{1,T}Q^{x}\!-\!A_{M} P\| \!+\! \lambda \|\bar {X}_{1,T}Q^{r}\!-\! B_{M} P\| \\
&\text{subject to }~~ \bar{X}_{0,T-1}Q^{x}=P, \label{eq:DD_matchingx4}\\
& \qquad \quad \quad~~~  \bar{X}_{0,T-1}Q^{r}=\mathbf{0}_{n}, \label{eq:DD_matchingr4}\\
&  \qquad \quad \quad~~~ 	\begin{bmatrix}
P & \bar{X}_{1,T}Q^{x}\\
\left(\bar{X}_{1,T}Q^{x}\right)' & P
\end{bmatrix} \succ \mathbf{0}_{2 n}, \label{eq:DD_stab_ave}
\end{align}
\end{subequations}	 
where $\lambda > 0$, and the two terms of the cost function are obtained by lifting the averaged matching conditions $\bar{X}_{1,T}Q^{x}=A_{M} P$ and $\bar{X}_{1,T}Q^{r}=B_{M}  P$. If a solution to \eqref{eq:relax1} is found, the control matrices are given by $K_x=\bar U_{0,T-1}Q^{x}P^{-1}$ and $K_r=\bar U_{0,T-1}Q^{r}P^{-1}$. We can term \eqref{eq:relax1} a \emph{certainty-equivalence} solution since the design is carried out as if the noise were zero, because of the use of averages. Note that, from \eqref{eq:relax1} we exactly recover Theorem \ref{thm:mismatch} as $N \rightarrow \infty$, since $\bar W_{0,T}=0$.
\begin{remark}[Features of the measurement noise]
	The assumption on the noise is not restrictive. Indeed, by properly detrending the measured states, one can satisfy the zero-mean hypothesis in practice. This preprocessing phase is shared by most identification techniques and state-of-the-art data-driven control approaches \cite{sysid}. \hfill $\blacksquare$ 
\end{remark}

Hereafter, we provide a stability result and some consideration related to the properties of \eqref{eq:relax1}. To this end, let us consider the following assumption.  
\begin{assumption} \label{ass:noise3}
	The data satisfy the following
	\begin{subequations}
		\begin{eqnarray} \label{ass:noise3_a}
		& \begin{bmatrix}
		\mathbf{0}_{m,T} \\  \bar V_{0, T-1}
		\end{bmatrix}
		\begin{bmatrix}
		\mathbf{0}_{m,T} \\  \bar V_{0, T-1}
		\end{bmatrix}'  \preceq \gamma_1 
		\begin{bmatrix}
		\bar U_{0,T-1} \\
		\bar X_{0, T-1}
		\end{bmatrix}
		\begin{bmatrix}
		\bar U_{0,T-1} \\
		\bar X_{0, T-1}
		\end{bmatrix}' \\[0.2cm]
		\label{ass:noise3_b}
		& \bar V_{1, T} \bar V_{1, T}' \preceq \gamma_2 \bar X_{1, T} \bar X_{1, T}'
		\end{eqnarray} 
		for some $\gamma_1 \in (0,0.5)$ and $\gamma_2 > 0$. 
	\end{subequations} \hfill $\square$
\end{assumption}
Since Assumption \ref{ass:noise3} requires the energy of the noise to be sufficiently small with respect to the one related to the data, it can be seen as a \emph{signal-to-noise ratio} condition. Under this hypothesis, we can provide sufficient conditions for closed-loop stability under noisy data as follows.

\begin{theorem} \label{thm:noisy}
	Consider system \eqref{eq:noiseless_model} and a reference model as in \eqref{eq:desired_behavior}. Suppose that Assumption \ref{ass:noise3} holds 
	and that \eqref{eq:relax1} is feasible. Let $(Q^x,Q^r,P)$ be any 
	solution and let $\alpha$ and 
	$\beta$ be positive constants such that
	\begin{equation}
	\Xi
	+\alpha \bar X_{1,T} (\bar X_{1,T})' \preceq 0, \quad M \preceq \beta I,
	\end{equation}
	where \vspace{-.1cm}
	\begin{eqnarray}\label{eq:DD_stab_average0}
	\Xi = \bar X_{1,T} M \bar X_{1,T}'-P, \quad
	M = Q^x P^{-1} (Q^x)'. \vspace{-.1cm}
	\end{eqnarray}
	If 
	\begin{eqnarray} \label{eq:stab_noise}
	\frac{6 \gamma_1+3\gamma_2}{1 - 2 \gamma_1} <
	\frac{\alpha^2}{2\beta (2\beta + \alpha)},
	\end{eqnarray}
	then $K_x= \bar U_{0,T-1}Q^{x}P^{-1}$ ensures closed-loop stability.
\end{theorem}

\begin{proof}
	The proof follows the same steps as the one of \cite[Theorem 5]{DePersis2020} and \cite[Corollary 1]{DePersis2020}, so we will omit most of the details. First of all, note that $\alpha,\beta >0$ exist since $\Psi \prec 0$ and $P \succ 0$ by hypothesis. 
	Further, under Assumption \ref{ass:noise3}, 
	\begin{equation}\label{eq:needed_torpove}
	\bar W_{0,T} \bar W_{0,T}' \preceq \gamma \bar X_{1,T} \bar X_{1,T}',
	\end{equation}
	with $\gamma$ equal to the left-hand side of \eqref{eq:stab_noise} \cite[Corollary 1]{DePersis2020}.
	Accordingly, we just need to prove that the matrix inequality in \eqref{eq:needed_torpove} with $\gamma < \alpha^2/(2\beta(2\beta+\alpha))$ ensures closed-loop stability. To this end, note that \eqref{eq:relax1} result in the closed-loop transition matrix $A+B K_x = (\bar X_{1,T}+\bar W_{0,T})Q^x P^{-1}$
	. Hence, $K_x$ ensures stability if 
	\begin{equation}\label{eq:DD_stab_average}
	(\bar X_{1,T}+\bar W_{0,T}) M 
	(\bar X_{1,T}+\bar W_{0,T})'-P
	\prec \mathbf{0}_{n}.
	\end{equation}
	This condition can be written as $\Xi+\Xi_1+\Xi_1'+\Xi_2 \prec \mathbf{0}_{n}$, with 
	$\Xi_1=\bar X_{1,T} M \bar W_{0,T}'$  and 
	$\Xi_2=\bar W_{0,T} M \bar W_{0,T}'$. 
	By applying 
	Young's matrix inequality, 
	a sufficient stability condition is thus given by
	$\Xi+\epsilon \bar X_{1,T} M \bar X_{1,T}'+
	(1+\epsilon^{-1}) \Xi_2
	\prec \mathbf{0}_{n}$
	with $\epsilon >0$ arbitrary. Because of the assumptions 
	on $\alpha$ and $\beta$, this latter condition holds if
	\begin{equation}\label{eq:DD_stab_average4}
	(\epsilon \beta -\alpha) \bar X_{1,T} \bar X_{1,T}'+
	\beta (1+\epsilon^{-1}) \bar W_{0,T} \bar W_{0,T}'
	\prec \mathbf{0}_{n}.
	\end{equation}
	By choosing $\epsilon=\alpha/(2\beta)$, we thus conclude that a 
	sufficient condition for stability is given by
	\begin{equation}\label{eq:DD_stab_average5}
	\bar W_{0,T} \bar W_{0,T}'
	\prec \frac{\alpha^2}{2\beta(2\beta+\alpha)} \bar X_{1,T} \bar X_{1,T}',
	\end{equation}
	which is satisfied because of \eqref{eq:needed_torpove} and 
	\eqref{eq:stab_noise}. \end{proof} 

\begin{remark}[On the choice of $\alpha$ and $\beta$]
From a practical standpoint, the parameters $\alpha$ and $\beta$ in Theorem \ref{thm:noisy} can be found by bisection. 	\hfill $\blacksquare$
\end{remark}

Theorem \ref{thm:noisy} makes no explicit use of the 
averaging strategy. Besides recovering  
Theorem \ref{thm:mismatch} as $N \rightarrow \infty$, this strategy 
plays a key role also for finite $N$, as remarked next.

\subsubsection{Hypothesis fulfillment} 
A first important motivation for considering the averaging strategy 
is related to some \emph{finite} sample properties of Gaussian noise. 
Consider 
$v \sim \mathcal {N} (0,\sigma^2 I_{n})$. 
As shown in \cite[Lemma 9]{deptLQR}, for any 
$\mu>0$ it holds that\footnote{An analogous bound holds for 
$\bar{V}_{1,T}$.}
\begin{eqnarray} \label{eq:Wainwright}
\|\bar V_{0,T-1}\| \leq \sigma \sqrt{\frac{T}{N}} \left( 1+ \mu + \sqrt{\frac{n}{T}} \right) \vspace{-.1cm}
\end{eqnarray}
with probability at least $1-e^{-T\mu^2/2}$. This inequality suggests that Assumption \ref{ass:noise3} is eventually satisfied if the experiments ensure hat $\bar U_{0,T-1}$, 
$\bar X_{0,T-1}$ and $\bar X_{1,T}$ do not vanish as $N$ increases. This is the case of  \emph{repeated} experiments, namely those ones that are carried out with the same input sequence $\mathcal U_{T-1}$ and from the same initial state $x(0)$. In this scenario, it is simple to see that for any probability level there
exists a finite $N$ such that Assumption \ref{ass:noise3} is satisfied
as long as $\mathcal U_{T-1}$ is persistently exciting, \emph{c.f.} \cite[Section 5.1]{deptLQR}. 

\subsubsection{Hypothesis verification} 

Given the data, the right-hand side of both \eqref{ass:noise3_a} and \eqref{ass:noise3_b} are known. Hence, Assumption~\ref{ass:noise3} can be checked from data 
whenever an upperbound on the noise is known deterministically. For Gaussian noise, \eqref{eq:Wainwright} makes it possible to have high-confidence bounds on the noise, hence on the probability that Assumption \ref{ass:noise3} holds. Note that, 
to have high confidence bounds, we need 
$T \mu^2$ large. Averaging permits 
to take 
$\mu^2$ large and compensate its effect on the estimate 
with $N$.

\subsubsection{Problem complexity}
Averaging allows us to exploit the benefits of large datasets, also required by systems identification and state-of-the-art data-driven control methods to cope with noisy data. Instead of running longer experiments, our strategy works by using several tests. This enables us to keep the number of constraints equal to the one in \eqref{eq:relax}, not increasing complexity of the data-driven problem. 

\begin{remark}[Increasing robustness]
	From Theorem \ref{thm:noisy} one sees that closed-loop stability becomes easier to 
	satisfy as $\beta$ gets smaller. This suggests to turn \eqref{eq:relax1} into an optimization problem where also $M$ is taken into account. Similarly to \cite{deptLQR}, this additional insight can be exploited by regularizing the objective in \eqref{eq:relax1} as follows
	\begin{eqnarray} \label{eq:robustness1}
	\begin{array}{rl}
	\displaystyle \underset{Q^x,Q^r,P}{\text{minimize}} &  
	\|\bar{\Psi}_{1}\| + \lambda \|\bar \Psi_2\|  + \lambda_1 \|M\| \\[0.1cm]
	\text{subject to} & \eqref{eq:DD_matchingx4}, \eqref{eq:DD_matchingr4},
	\eqref{eq:DD_stab_ave}
	\end{array}
	\end{eqnarray}	 
	with $\bar{\Psi}_{1}=\bar {X}_{1,T}Q^{x}\!-\!A_{M} P$, $\bar{\Psi}_2=\bar {X}_{1,T}Q^{r}\!-\! B_{M} P$ and $\lambda,\lambda_1>0$, which is again a semi-definite program. \hfill $\blacksquare$
\end{remark}

%

\section{Numerical examples}\label{Sec:examples}
We now assess the performance achieved when tackling Problem~\ref{pb:DDCstatic_problem} through the certainty-equivalence solution in 
\eqref{eq:relax1} by minimizing L1-norms, with $\lambda=1$. We consider two third-order systems with three inputs, one open-loop stable and the other unstable. Our goal is to decouple the dynamics of each state and enforce closed-loop stability. The examples have been designed so that the matching conditions in \eqref{eq:MB_matching} are fulfilled by a unique pair of optimal gains $K_{x}^{\star}$ and $K_{r}^{\star}$. This enables us to quantitatively evaluate the quality of the estimated controllers 
through \vspace{-.1cm}
\begin{equation}\label{eq:quality_evaluate}
\|K_{x}-K_{x}^{\star}\|_{2}, \qquad \quad \|K_{r}-K_{r}^{\star}\|_{2},\vspace{-.1cm}
\end{equation}
that have to be small for closed-loop matching to be attained.

Within each scenario, the performance is assessed by carrying out $100$ Monte Carlo simulations for increasing levels of noise, which is quantified via the \emph{Signal-to-Noise Ratio} (SNR) over the three output channels, \emph{i.e.,}
 \begin{equation}
 \mbox{SNR}_{j}=10 \log\left(\frac{\sum_{t=0}^{T}(x_{j}^{\mathrm{o}}(t))^{2}}{\sum_{t=0}^{T}(v_{j}^{(i)}(t))^{2}}\right)\!\!,~[dB] 
\end{equation} 
with $j\!=\!1,2,3$ and $i\!=\!1,\ldots,N$. Independently of the framework, the sequence $\{v(t)\}_{t=0}^{T}$ is zero-mean and Gaussian distributed, \emph{i.e.,} $v \sim \mathcal{N}(0,\sigma^{2}I)$, with increasing variance so as to span the interval $\mbox{SNR} \in [3,100]$~dB. By considering initial datasets of length $T=30$ samples, for each Monte Carlo run we evaluate the matching performance for an growing number $N$ of experiments, starting from $N=1$ up to $N=1000$. These conditions correspond to datasets comprising a minimum of $30$ state samples, up to $30 \cdot 10^{3}$ state points. The design problem was solved with the CVX package \cite{cvx}, by imposing
\begin{equation*}
\begin{bmatrix}
P & \bar{X}_{1,T}Q^{x}\\
\left(\bar{X}_{1,T}Q^{x}\right)' & P
\end{bmatrix} \succeq 10^{-10} \cdot I_{2\cdot n}.
\end{equation*}


\subsection{Model-reference control of a stable system}
\begin{figure*}[!tb]
	\centering
	\begin{tabular}{cc}
		\subfigure[$\|K_{x}-K_{x}^{\star}\|_{2}$ \emph{vs} average SNR over states]{\includegraphics[scale=.5,trim=0cm 2cm 4cm 17cm,clip]{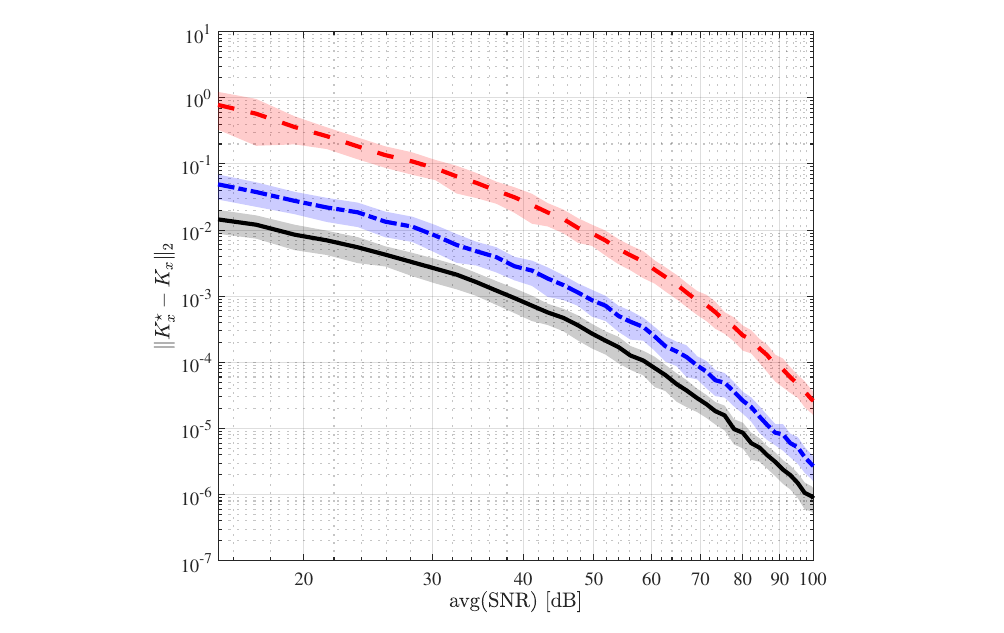}} \hspace{1.5cm}&\hspace{1.5cm} \subfigure[$\|K_{r}-K_{r}^{\star}\|_{2}$ \emph{vs} average SNR over states]{\includegraphics[scale=.5,trim=0cm 2cm 4cm 17cm,clip]{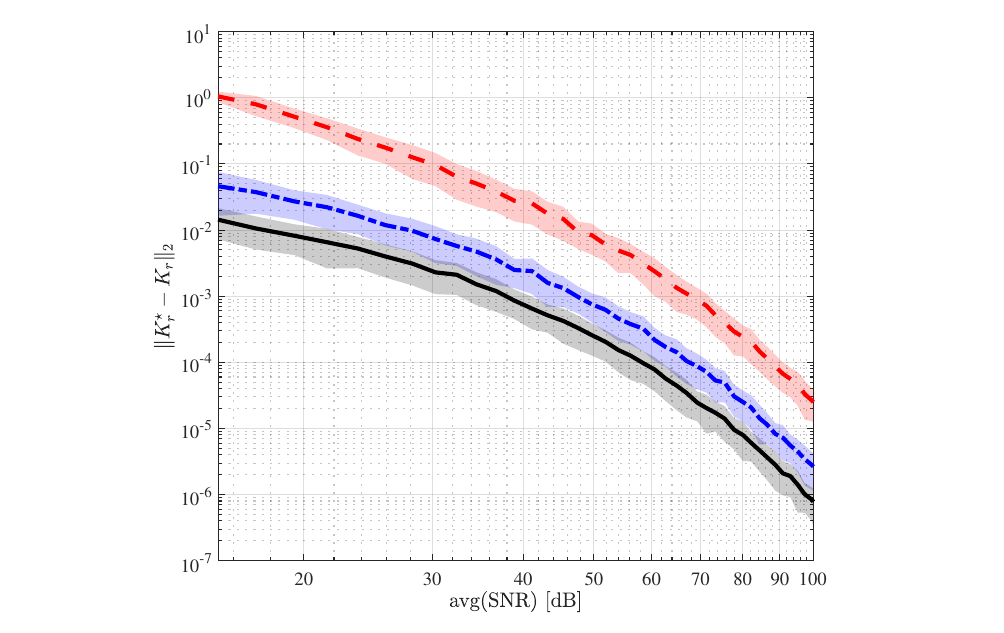}}
	\end{tabular}\vspace{-.1cm}
	\caption{Stable system: Average (solid lines) and standard deviation (colored areas) of the quality indexes in \eqref{eq:quality_evaluate} over the Monte Carlo runs resulting in a stable closed-loop. The results in red are referred to $N=1$, the blue ones are related to $N=100$, while the black results are associated with $N=1000$.}\label{Fig:Stab_comparison}
\end{figure*}
\begin{figure*}[!tb]
	\centering
	\begin{tabular}{ccc}
		\subfigure[First state component]{\includegraphics[scale=.475,trim=1.5cm 0.75cm 8cm 20cm,clip]{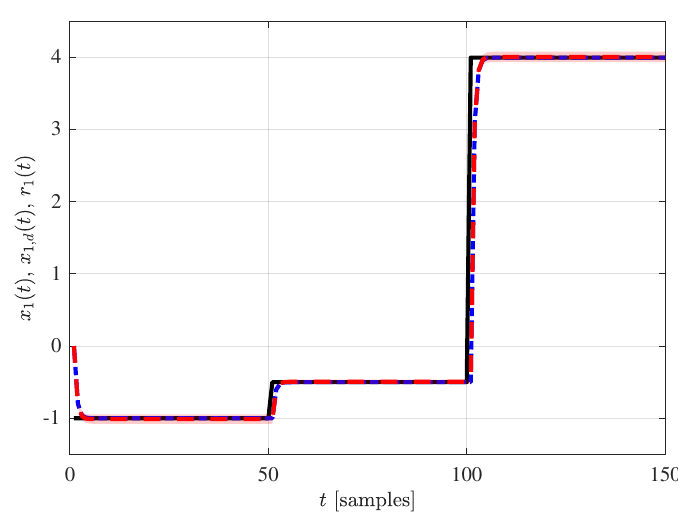}} &
		\subfigure[Second state component]{\includegraphics[scale=.475,trim=1.5cm 0.75cm 8cm 20cm,clip]{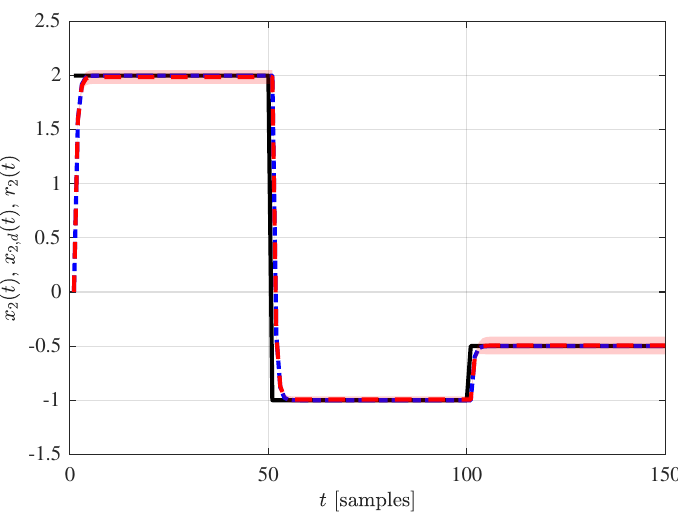}} & \subfigure[Third state component]{\includegraphics[scale=.475,trim=1.5cm 0.75cm 7cm 20cm,clip]{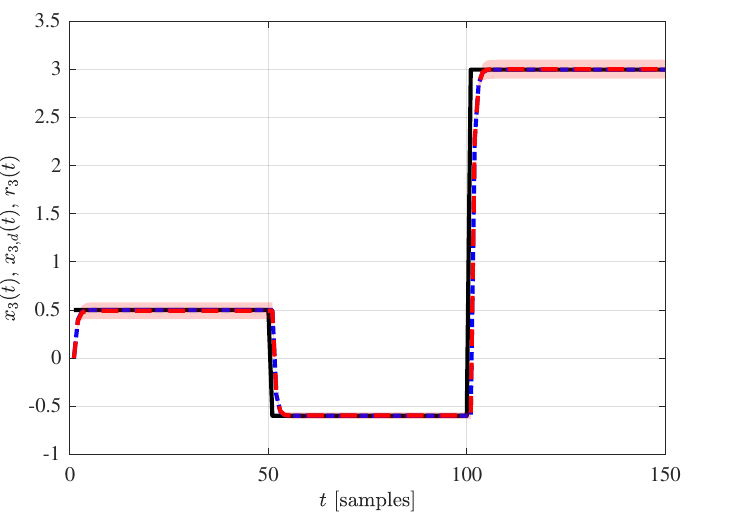}} 
	\end{tabular}\vspace{-.1cm}
	\caption{Stable system: reference (solid black), desired response (dotted dashed blue) \emph{vs} average (dashed red) closed-loop output for $N=100$. The red area indicates the standard deviation with respect to the average response, which is almost negligible when looking at the first state component.}\label{Fig:stab_traj}
\end{figure*}
We initially consider a randomly generated open-loop stable system of the form in \eqref{eq:noiseless_model}, characterized by the matrices
\begin{equation}
\begin{aligned}
A&=\begin{bmatrix}
0.1344  &  0.2155 &  -0.1084\\
0.4585  &  0.0797 &   0.0857\\
-0.5647 &  -0.3269  &  0.8946
\end{bmatrix},\\
B&=\begin{bmatrix}
0.9298  &  0.9143 &  -0.7162\\
-0.6848  & -0.0292 &  -0.1565\\
0.9412  & 0.6006  &  0.8315
\end{bmatrix},
\end{aligned}
\end{equation}
with 
$A$ having three distinct real eigenvalues $\{0.9536,\!-0.2118,0.3670\}$. Our aim is to design a static law as in \eqref{eq:noisy_law} such that the behavior of the reference model in \eqref{eq:desired_behavior} with $A_{M}=0.2 \cdot I_{3}$ and $B_{M}=0.8 \cdot I_{3}$ 
is matched. This choice allows us to speed up the dynamics of the open-loop system and to attain zero steady-state error when tracking step like references. Perfect matching is achieved with 
\begin{equation}
\begin{aligned}
K_{x}^{\star}&=\begin{bmatrix}
0.6308 &  -0.2920  &  0.3080\\
-0.3814  &  0.4011 &  -0.7166\\
0.2405  &  0.4340  & -0.6664
\end{bmatrix}, \\
K_{r}^{\star}&=\begin{bmatrix}
0.0768 &  -1.3126 &  -0.1809\\
0.4654  &  1.5957  &  0.7012\\
-0.4231  &  0.3332  &  0.6604
\end{bmatrix}.
\end{aligned}
\end{equation}
To retrieve them from data, the system is fed with a random input sequence uniformly distributed within $[-2,2]$, which guarantees that the plant is persistently excited. 

The obtained quality indexes (see \eqref{eq:quality_evaluate}) are shown in \figurename{~\ref{Fig:Stab_comparison}}. As expected, the lower the noise corrupting the state measurements is, the better the quality of the retrieved gains results. Moreover, an increasing number $N$ of experiments leads to a considerable reduction in the error between the optimal and data-driven gains, as we get closer to the limit condition. Indeed, the plant is never destabilized by the controller when $N \geq 100$, even for an average SNR $\approx 4$~dB, while the attained closed-loop is always stable for $N=2$ whenever the average SNR over the states is above $11$~dB. 

We then assess the actual matching performance, by comparing the desired and achieved closed-loop output for 
a piecewise constant reference. This comparison is shown in \figurename{~\ref{Fig:stab_traj}}, where we focus on the set of controllers retrieved when the data are characterized by an average SNR $\in [11.18, 14.27]$~dB. Clearly, the desired and closed-loop behavior closely match 
on average and the variance over the Monte Carlo runs is almost negligible on transients. 
A slight mismatch occurs at steady-state, when we do not 
exactly attain 
a unitary closed-loop DC gain. 

\subsection{Model-reference control of an unstable system}
\begin{figure*}[!tb]
	\centering
	\begin{tabular}{cc}
		\subfigure[$\|K_{x}-K_{x}^{\star}\|_{2}$ \emph{vs} average SNR over states]{\includegraphics[scale=.5,trim=0cm 2cm 4cm 17cm,clip]{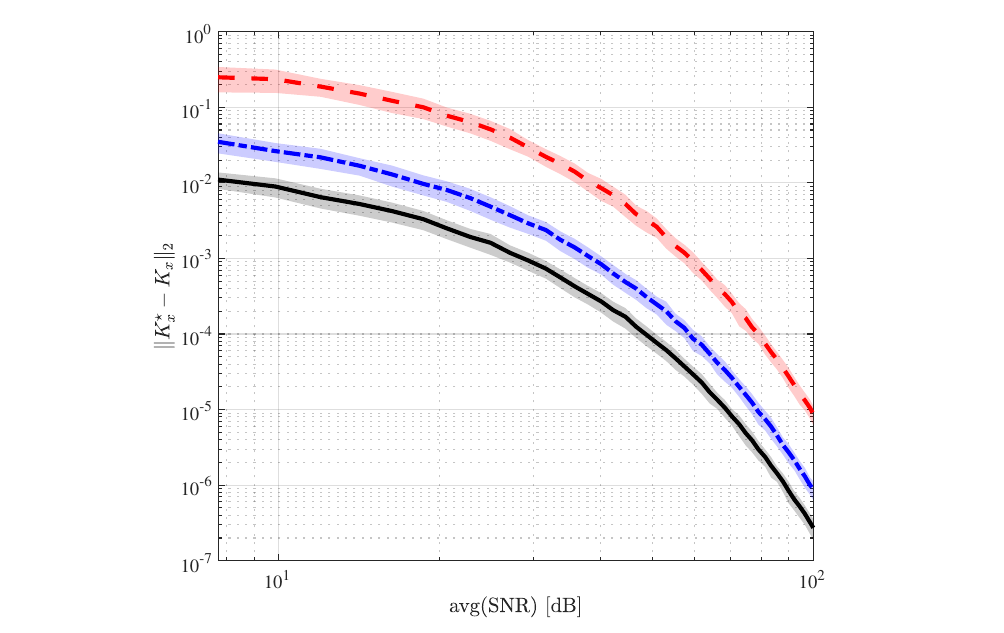}} \hspace{1.5cm}&\hspace{1.5cm} \subfigure[$\|K_{r}-K_{r}^{\star}\|_{2}$ \emph{vs} average SNR over states]{\includegraphics[scale=.5,trim=0cm 2cm 4cm 17cm,clip]{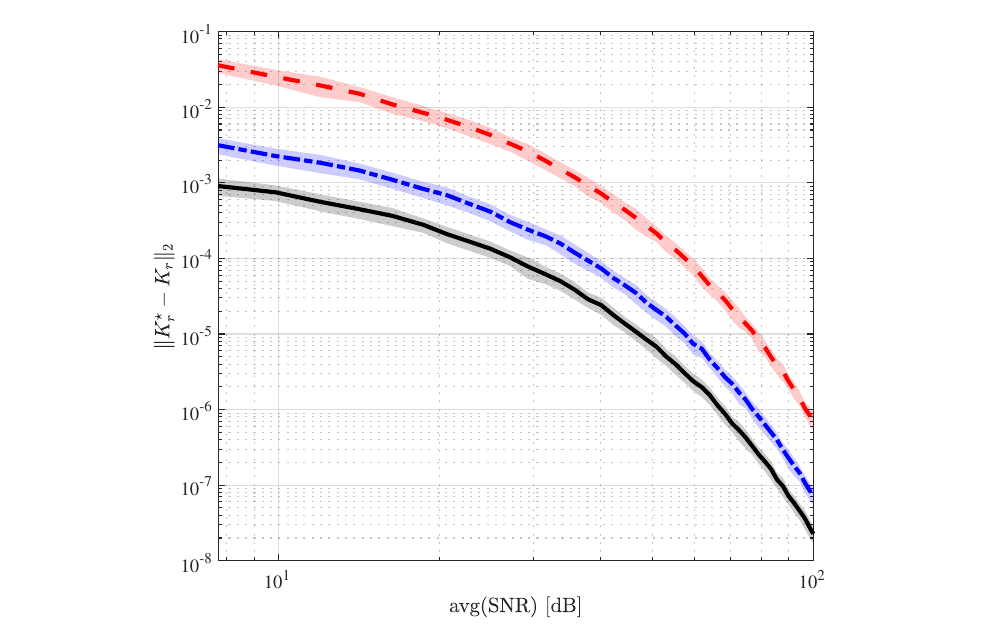}}
	\end{tabular}\vspace{-.1cm}
	\caption{Unstable system: Average (solid lines) and standard deviation (colored areas) of the quality indexes in \eqref{eq:quality_evaluate} over the Monte Carlo runs resulting in a stable closed-loop. The results in red are referred to $N\!=\!1$, the blue ones are related to $N=100$, while the black results are associated with $N=1000$.}\label{Fig:Unstab_comparison}
\end{figure*}

\begin{figure*}[!tb]
	\centering
	\begin{tabular}{ccc}
		\subfigure[First state component]{\includegraphics[scale=.475,trim=1.5cm 0.75cm 8cm 20cm,clip]{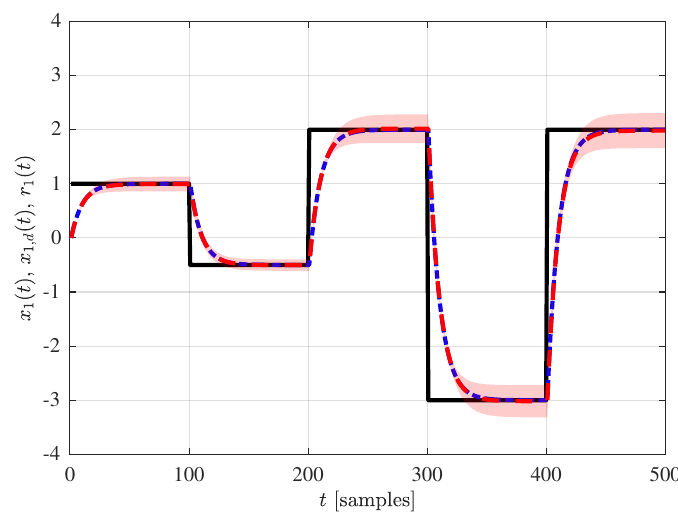}} &
		\subfigure[Second state component]{\includegraphics[scale=.475,trim=1.5cm 0.75cm 8cm 20cm,clip]{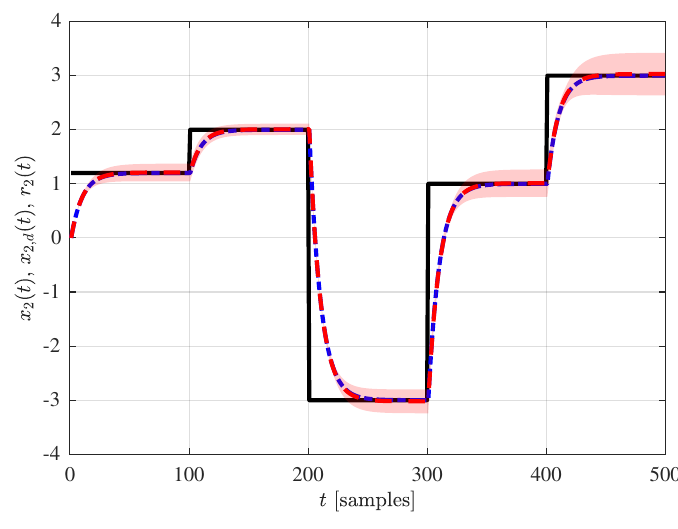}} & \subfigure[Third state component]{\includegraphics[scale=.475,trim=1.5cm 0.75cm 7cm 20cm,clip]{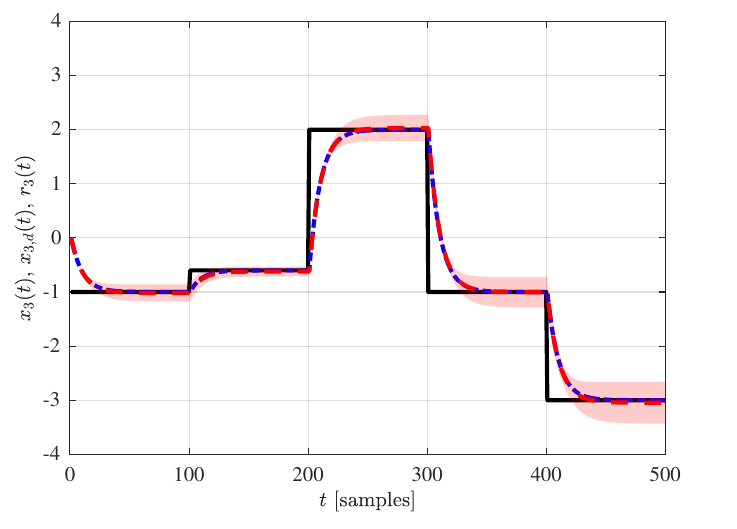}} 
	\end{tabular}\vspace{-.1cm}
\caption{Unstable system: reference (solid black), desired response (dotted dashed blue) \emph{vs} average (dashed red) closed-loop output for $N=100$. The red area indicates the standard deviation with respect to the average response.}\label{Fig:unstab_traj}
\end{figure*}
\begin{table}[!tb]
	\caption{Unstable closed-loop \emph{vs} average SNR and $N$.}\label{Tab:1}
	\centering
	\begin{tabular}{clc}
		\multicolumn{2}{c}{} & \multicolumn{1}{c}{\# unstable instances}\\
		\cline{2-3}
		\multirow{3}{.2\textwidth}{avg(SNR) $\!\in\! [14.12, 17.68]$~dB} & $N=1$ & 17\\
		\cline{2-3}
		\multicolumn{1}{c}{} & $N=2$ & 4 \\
		\cline{2-3}
		\multicolumn{1}{c}{} & $N=100$ & 0 \\
		\hline 
		\multirow{3}{.18\textwidth}{avg(SNR) $\!\in\! [6.08, 9.33]$~dB} & $N=1$ & 65\\
		\cline{2-3}
		\multicolumn{1}{c}{} & $N=2$ & 48\\
		\cline{2-3}
		\multicolumn{1}{c}{} & $N=100$ & 0\\
		\hline
	\end{tabular}
\end{table}
Let us now consider the benchmark unstable system introduced in \cite{Recht2020}, with
\begin{equation}
A=\begin{bmatrix}
1.01 & 0.01 & 0\\
0.01 & 1.01 & 0.01\\
0 & 0.01 & 1.01
\end{bmatrix}, \quad B=I_{3}.
\end{equation}
In this case, the matrices of the reference model in \eqref{eq:desired_behavior} are chosen as $A_{M}=0.9 \cdot I_{3}$ and $B_{M}=0.1 \cdot I_{3},$ 
so as to dictate a stable behavior and guarantee that step-references are perfectly tracked. This reference model is matched with the gains \vspace{-.1cm}
\begin{equation}
K_{x}^{\star}=\begin{bmatrix}
-0.11  & -0.01 & 0\\
-0.01 &  -0.11 & -0.01\\
0 & -0.01 &  -0.11
\end{bmatrix}, \quad K_{r}^{\star}=0.1 \cdot I_{3}.
\end{equation}
Due to the open-loop instability of the plant, experiments are carried out in closed-loop by stabilizing the system with a static controller of the same form as \eqref{eq:noisy_law}, with $K_{x}=-I_{n}$ and $K_{r}=I_{n}$. The reference to be tracked throughout the experiments is selected as a sequence of uniformly distributed samples within $[-5,10]$ generated at random. This choice guarantees that the input fed to the system is persistently exciting.

\figurename{~\ref{Fig:Unstab_comparison}} reports the quality indexes in \eqref{eq:quality_evaluate} for an increasing level of noise and different dimensions of the dataset. As for the stable plant, larger sets of data lead to a better reconstruction of the matching gains, as the average formulation in \eqref{eq:relax1} increasingly resemble the noiseless ones. This result is aligned with the reduction in the number of unstable instances obtained when more experiments are performed, as shown in \tablename{~\ref{Tab:1}}. Indeed, their number is consistently reduced as long as $N$ increases, becoming zero when the data gathered over $100$ experiments are used. 
The retrieved 
controllers are always stabilizing for datasets yielding 
average SNRs above $20$~dB, independently of $N$. 



Let us now focus on the performance achieved when the training data are characterized by an average SNR $\in [10.33, 13.53]$~dB. As shown in \figurename{~\ref{Fig:unstab_traj}}, the desired behavior and the attained one match exactly on average, when a piecewise constant reference is considered. While the transient behavior is generally not affected by differences in the realization of the dataset, there is a variation in the steady state values of the closed-loop output over the $100$ gains retrieved. This behavior can be associated with slight differences between the optimal and actual gains, that eventually lead to a steady-state offset. 
This problem can be handled by exploiting an integrator, which we aim at introducing in future works.  

\section{Conclusions} \label{sec:conclusions}
In this work, a direct data-driven design strategy for model-reference 
linear controllers guaranteeing Lyapunov stability has been introduced and 
discussed. Unlike existing data-based approaches, the one proposed here 
allows us to obtain non-asymptotic stability guarantees
also in case of unfeasible perfect matching and noisy measurements.
Numerical studies have confirmed the effectiveness of the 
proposed strategy.

Future work will also be devoted to the analysis of scenarios where the order of the reference model does not coincide with that of the plant 
and to further experimental validation, especially considering high-order systems. Moreover, we will explore approaches for the automatic selection of the reference model based on soft specifications provided by the user, similarly to \cite{BRESCHI2021bis}.

\bibliographystyle{plain}
\bibliography{ModelReferenceDDC_v9.bib}


\end{document}